\documentclass[final]{l4dc2020}
\usepackage{mathematix}
\usepackage{xfrac}
\usepackage{etoolbox}
\usepackage{enumitem}

\DeclareMathAlphabet\mathbfcal{OMS}{cmsy}{b}{n}

\declarecommand{\QED}{\hfill\ensuremath{\blacksquare}}%

\declarecommand{\covarmhalf}{\covar^{-\nicefrac{1}{2}}}
\declarecommand{\covarhalf}{\covar^{\nicefrac{1}{2}}}
\declarecommand{\Ihat}{\est{I}}
\declarecommand{\nsample}{M}
\declarecommand{\natseq}[2]{\N_{[#1,#2]}} 
\declarecommand{\posN}{{\N_+}}

\declarecommand{\ssum}{\textstyle \sum}
\declarecommand{\refinenv}[1]{}

\declarecommand{\dodd}[1]{#1}
\declarecommand{\dd}[1]{{\hat{#1}}}
\declarecommand{\est}[1]{{\hat{#1}}}

\declarecommand{\rnd}[1]{{#1}}

\renewcommand{\vec}[1]{{\mathbf{#1}}}
\renewcommand{\amb}{{\dd{\mathcal{A}}}}

\declarecommand{\xd}{{n_x}}
\declarecommand{\ud}{{n_u}}
\declarecommand{\xs}{\rnd{x}}
\declarecommand{\xinits}{\bar{\xs}}
\declarecommand{\xinitsr}{\tilde{\xs}}
\declarecommand{\us}{\rnd{u}}
\declarecommand{\muls}{\rnd{w}}
\declarecommand{\mulsdd}{\dd{\muls}}
\declarecommand{\mulse}{\rnd{v}}
\declarecommand{\muld}{{n_w}}
\declarecommand{\mulpdf}{\mathit{F}}
\declarecommand{\net}{\mathcal{N}}

\declarecommand{\sequence}[1]{\mathbf{#1}}

\declarecommand{\ssg}{\Omega}
\declarecommand{\sag}{\mathcal{F}}
\declarecommand{\probg}{\mathbb{P}}
\declarecommand{\ssmul}{\Omega}
\declarecommand{\samul}{\mathcal{F}}
\declarecommand{\probmul}{{\prob_\muls}}
\declarecommand{\ssmulp}{{\ssmul^\infty}}
\declarecommand{\samulp}{{\samul^\infty}}
\declarecommand{\probmulp}{{\probmul^\infty}}
\declarecommand{\ssmuls}{{\ssmul^\nsample}}
\declarecommand{\samuls}{{\samul^\nsample}}
\declarecommand{\probmuls}{{\probg}}
\declarecommand{\ssmule}{{\ssmul}}
\declarecommand{\samule}{{\samul}}
\declarecommand{\probmule}{{\prob_\mulse}}
\declarecommand{\ssmulep}{\ssmulp}
\declarecommand{\samulep}{\samulp}
\declarecommand{\probmulep}{{\probmule^\infty}}
\declarecommand{\ssxinit}{E}
\declarecommand{\saxinit}{\mathcal{E}}
\declarecommand{\probxinit}{\mathbb{G}}
\declarecommand{\Es}[3]{\E} 
\declarecommand{\Er}[2]{\E} 
\declarecommand{\Ere}[2]{\E} 
\declarecommand{\Ese}[3]{\E} 

\declarecommand{\mean}{\mu}
\declarecommand{\covar}{\Sigma}
\declarecommand{\covare}{\covar_0}
\declarecommand{\covardr}{\dd{\covar}_{\mathrm{dr}}}

\declarecommand{\rmean}{r_{\mean}(\beta)}
\declarecommand{\rcovar}{r_{\scriptscriptstyle{\covar}}(\beta)}

\declarecommand{\tmean}{t_{\mean}}
\declarecommand{\tcovar}{t_{\scriptscriptstyle{\covar}}}

\declarecommand{\partsim}{{N}}

\declarecommand{\Val}{V}
\declarecommand{\ie}{\emph{i.e.,}}
\declarecommand{\eg}{\emph{e.g.,}}
\declarecommand{\iff}{iff}
\declarecommand{\iid}{\emph{i.i.d.}}
\declarecommand{\mss}{\emph{m.s.s.}}
\declarecommand{\emss}{\emph{e.m.s.s.}}
\declarecommand{\sdp}{SDP}

\declarecommand{\sgeq}{\succeq} 
\declarecommand{\sgt}{\succ} 
\declarecommand{\sleq}{\preceq} 

\DeclareMathOperator{\subG}{subG}

\newlist{inlinelist}{enumerate*}{1}
\setlist*[inlinelist,1]{%
  label=(\roman*),
}
\newlist{proofsteps}{enumerate}{1}
\setlist[proofsteps,1]{%
  label=\bfseries{\Roman*},
  itemindent=0pt,
  wide =0.5\parindent,
  listparindent=0pt,%
  afterlabel={{.\nobreakspace}}
}
\newlist{theoremcontents}{enumerate}{1}
\setlist[theoremcontents,1]{%
  label=\bfseries{\Roman*},
  itemindent=0pt,
  wide=0.5\parindent,
  itemsep=0.0pt,
  parsep=0.0pt,
  listparindent=-10pt,%
  afterlabel={{.\nobreakspace}}
}

\makeatletter
\newcommand*{\inlineequation}[2][]{%
  \begingroup
    \refstepcounter{equation}%
    \ifx\\#1\\%
    \else
      \label{#1}%
    \fi
    \relpenalty=10000 %
    \binoppenalty=10000 %
    \ensuremath{%
      #2%
    }%
    ~~\@eqnnum
  \endgroup
}
\makeatother

\usepackage{caption}
\usepackage{amssymb}

\newcommand{\preliminary}[1]{{\leavevmode\color{Cornsilk4} #1}}
\renewcommand{\preliminary}[1]{}
\newcommand{\todo}[3]{\preliminary{{\color{#1} \small [TODO] \textbf{#2} --- #3}}}
\newcommand{\mathijs}[1]{\todo{DarkOliveGreen3}{Mathijs}{#1}}

\def\isarchiv{1}
\ifdefined\isarchiv
\newcommand{\archiv}[2]{{\leavevmode\color{black!100} #1}}
\else
\newcommand{\archiv}[2]{{\leavevmode\color{black!100} #2}}
\fi
\usepackage{stmaryrd}
\usepackage{cleveref}

\newtheorem{assumption}{Assumption}

\makeatletter
\newcommand{\customlabel}[2]{%
\protected@write \@auxout {}{\string \newlabel {#1}{{#2}{\thepage}{#2}{#1}{}} }%
\hypertarget{#1}{#2}
}
\makeatother

\SetSymbolFont{stmry}{bold}{U}{stmry}{m}{n}

\title[Data-driven distributionally robust LQR]{Data-driven distributionally robust LQR with multiplicative noise} 
\usepackage{times}

\author{\Name{Peter Coppens} \Email{peter.coppens@kuleuven.be}\\
\Name{Mathijs Schuurmans} \Email{mathijs.schuurmans@kuleuven.be}\\
\Name{Panagiotis Patrinos} \Email{panos.patrinos@kuleuven.be}\\
\addr STADIUS, Department of Electrical Engineering (ESAT),
KU Leuven, Leuven, Belgium.}

\begin{document}

\maketitle

\begin{abstract}%
  We present a data-driven method for solving the linear quadratic regulator problem for systems with multiplicative disturbances, the distribution of which is only known through sample estimates. We adopt a distributionally robust approach to cast the controller synthesis problem as semidefinite programs. Using results from high dimensional statistics, the proposed methodology ensures that their solution provides mean-square stabilizing controllers with high probability even for low sample sizes. As the sample size increases the closed-loop cost approaches that of the optimal controller produced when the distribution is known. We demonstrate the practical applicability and performance of the method through a numerical experiment. 
\end{abstract}


\section{Introduction}
We will develop controllers for linear systems with time-varying parametric uncertainty, which may cover a wide range of system classes extensively studied in the literature. For example, we obtain Linear Parameter Varying (LPV) systems when the disturbance is observable at each time step \citep{Wu1996,Byrnes1979}, Linear Difference Inclusions (LDIs) when it is unknown but norm-bounded \citep{Boyd1994} and stochastic systems with multiplicative noise when it varies stochastically \citep{Wonham1967}.

In many practical applications, however, the distribution of the disturbance is not known. These traditional control approaches either make a boundedness assumption on the disturbance or on its moments, which allows for a fully robust approach \citep{ElGhaoui1995}. 
Such approaches, however, disregard any statistical information that may be obtained on the distribution of the disturbances. Our aim, instead, is to design linear controllers which use sampled data to improve performance over fully robust approaches, while inheriting many of the system-theoretical guarantees of a robust control strategy. 
To this end, we adopt a distributionally robust (DR) approach \citep{dupacova_minimax_1987,Delage2010} towards solving the infinite-horizon \emph{Linear Quadratic Regulator} (LQR) problem, where we minimize the expected cost for the worst-case distribution in a so-called \emph{ambiguity set} computed based on the available data such that it contains the true distribution with high probability. 
Such a DR approach addresses most of the difficulties associated with learning automatically, since the ambiguity set directly models the uncertainty in the sample-based estimates against which the controllers will be robust.
Similar techniques were recently studied in \cite{Schuurmans2019safe} for stochastic jump linear systems and in \cite{Dean2019} for deterministic systems, where the system matrices $A$ and $B$ are learned from data. 

The work of \cite{Gravell2019} also deals with learning control of linear systems with multiplicative noise, albeit from a different perspective. They employ a policy gradient algorithm, which requires the initial guess for the control gain to be stabilizing. By contrast, obtaining such a controller is the main goal of our approach. 

Our main contributions are summarized as follows. Leveraging recent results from high dimensional statistics, we provide practical high-probability confidence bounds for the ambiguity sets, which depend only on known quantities (\Cref{sec:bounds}). We then extend the solution of the (\emph{nominal}) infinite horizon LQR problem with known distribution to related DR counterparts which account for the ambiguity on the disturbance distribution. Whenever the mean of the disturbance is known, we show that the DR problem is equivalent to a semidefinite program (\sdp{}) which has the same form as the nominal one. 
Next, we extend the formulation to the setting in which both the mean and the covariance are only known to lie in an ellipsoidal ambiguity set (\Cref{sec:mean}), for which we can only approximate the optimal controller.

 \archiv{}{Due to space limitations, only sketches of the proofs are provided here. We refer the reader to the technical report \cite{Coppens2019}\mathijs{Add identifier to bibtex entry} for the full versions of the proofs. }

\subsection{Notation}
Let $\Re$ denote the reals\archiv{,}{{ }and} $\N$ the naturals\archiv{{ }and $\posN = \N \setminus \{0\}$}{}. For symmetric matrices $P, Q$ we write $P \sgt Q$ ($P \sgeq Q$) to signify that $P - Q$ is positive (semi)definite and denote by $\otimes$ the Kronecker product. We assume that all random variables are defined on a probability space $(\Omega, \mathcal{F}, \prob)$, with $\Omega$ the sample space, $\mathcal{F}$ its associated $\sigma$-algebra and $\prob$ the probability measure. 
Let $\rnd{y} : \Omega \rightarrow \Re^n$ be a random vector defined on $(\Omega, \mathcal{F}, \prob)$. With some abuse of notation we will write $y \in \mathbb{R}^n$ to state the dimension of this random vector. Let $\prob_y$ denote the distribution of $y$, \ie{} $\prob_y(A) = \prob[y \in A]$. Then, a trajectory $\{y_i\}_{i = 1}^{N}$ of identically and independently distributed (\iid) copies of $\rnd{y}$ is defined by the distribution it induces.
That is, for any $A_0, \ldots, A_N \in \mathcal{F}$ we define $\prob_{y_{0}, \ldots, y_{N}}(A_0 \times \dots \times A_N) \dfn \prob[y_0 \in A_0 \land \dots \land y_N \in A_N] = \prod_{i=0}^{N} \prob_y(A_i)$. This definition can be extended to infinite trajectories $\{\rnd{y}_i\}_{i\in \N}$ by Kolmogorov's existence theorem \cite{Billingsley1995}.
We will write the expectation operator as $\E$. \archiv{We denote by $\E\left[ y \mid z \right]$ the conditional expectation with respect to $z$. }{}Let
$\mathcal{M}$ denote the set of probability measures defined on $(\Re^{\muld}, \B)$, with $\B$ the Borel $\sigma$-algebra of $\Re^\muld$.

\section{Problem statement} \label{sec:nominal}
Consider the stochastic discrete-time system with input- and state-multiplicative noise given by:
\begin{align}
  \xs_{k+1} &= A(\muls_k)\xs_k+B(\muls_k) \us_k \label{eq:sys1} \\
\textrm{with}\qquad\ A(\muls)&\dfn A_0+ \sum_{i=1}^\muld \muls^{(i)}A_i,\quad B(\muls)\dfn B_0+ \sum_{i=1}^\muld \muls^{(i)}B_i,\nonumber
\end{align}
where at each time $k$, $\xs_k \in \mathbb{R}^\xd$ denotes the state, $\us_k \in \mathbb{R}^\ud$ the input and $w_k \in \mathbb{R}^\muld$ an \iid{} copy of a square integrable random variable $w$ distributed according to $\probmul$. Let $w^{(i)}$ denote the $i$-th element of vector $w$. We introduce the following shorthands: $\vec{A} \dfn \trans{\smallmat{\trans{A}_1 & \ldots & \trans{A}_\muld}}$, 
$\vec{B} \dfn \trans{\smallmat{ \trans{B}_1 & \ldots & \trans{B}_\muld}}$, $\vec{A}_0 \dfn \trans{\smallmat{ \trans{A}_0 & \trans{\vec{A}} }}$, $\vec{B}_0 \dfn \trans{\smallmat{ \trans{B}_0 & \trans{\vec{B}}}}$.
We also define $\covare \dfn \smallmat{1 & \trans{\mean} \\ \mean & \covar + \mean \trans{\mean}}$, where $\mean \dfn \Er{\muls}{\probmul}\left[ \muls \right]$, $\covar \dfn \Er{\muls}{\probmul} \left[ (\muls - \mu) \trans{(\muls - \mu)} \right]$. 

\subsection{Nominal stochastic LQR problem and solution}
The primary goal is to solve the following LQR problem:
\begin{equation} \label{eq:lqr1}
  \begin{aligned}
    & \minimize_{\us_0, \us_1, \ldots} & & \Es{\muls}{\probmul}{\infty} \left[ \sum_{k=0}^{\infty} \trans{\xs_k} Q \xs_k + \trans{\us_k} R \us_k \right] \\
    & \stt & & \xs_{k+1} = A(\muls_k)\xs_k+B(\muls_k)\us_k, \quad k\in\N \\
    &&& \xs_0 = \xinits,
  \end{aligned}
\end{equation}
where we assume that $Q \sgt 0$ and $R \sgt 0$. The solution of 
\eqref{eq:lqr1} will yield a controller that renders the closed-loop system exponentially stable in the mean square sense, which is defined as follows.
\begin{definition}[(Exponential) Mean Square Stability]
  We say that an autonomous system $x_{k+1} = A(w_k) x_k$ is mean square stable (\mss) \iff{} $\Es{w}{\probmul}{k} \left[ \trans{\xs_k} \xs_k\right] \rightarrow 0$ as $k\rightarrow \infty$. It is exponentially mean square stable (\emss) \iff{} there exists a pair of positive constants $\gamma \in \left(0, \, 1\right)$ and $c$ such that $\Es{w}{\probmul}{k}\left[ \trans{\xs_k} \xs_k \right] \leq c\gamma^k \nrm{x_0}$ for all $k \in \N$ and for each $x_0 \in \Re^\xd$. 
\end{definition}
This property can be verified using the classical Lyapunov operator \citep{Morozan1983}:
\begin{theorem}[Lyapunov stability] \label{the:lyapunov}
  For the autonomous system $x_{k+1} = A(w_k) x_k$ the following statements are then equivalent:
  \begin{inlinelist}
    \item it is \mss,
    \item it is \emss,
    \item $\exists P \sgt 0$:
  \end{inlinelist}
    \begin{equation} \label{eq:lyapunov}
      P - \trans{\vec{A}_0}(\covare \otimes P)\vec{A}_0 \sgt 0.
    \end{equation}
\end{theorem}
\archiv{
\begin{proof}
  See Appendix~\ref{app:lyapunov}.
\end{proof}
}

\sloppy The LQR problem \eqref{eq:lqr1} has been studied for many variations of \eqref{eq:sys1} \citep{Morozan1983, Costa1997}. The following proposition is then similar to many classical results in literature:
\begin{proposition} \label{the:nominal}
  Consider a system with dynamics \eqref{eq:sys1} and the associated LQR problem \eqref{eq:lqr1}. Assuming that \eqref{eq:sys1} is mean square stabilizable, \ie{} there exists a $K$ and $P \sgt 0$ such that \eqref{eq:lyapunov} holds for the closed-loop system $\xs_{k+1} = (A(\muls_k) + B(\muls_k) K)\xs_k$,
  then the following statements hold.
  \begin{theoremcontents} 
    \item \label{the:nominal:first} The optimal solution of \eqref{eq:lqr1} is given by $K_{\infty} = -(R + G(P_\infty))^{-1}H(P_\infty)$, with $P_{\infty}$ the solution of the following Riccati equation:
      \begin{equation} \label{eq:nominal_riccati}
        P_{\infty} = Q + F(P_{\infty}) - \trans{H(P_\infty)}(R + G(P_\infty))^{-1}H(P),
      \end{equation}
      where \sloppy
      $F(P) \dfn \trans{\vec{A}_0} (\covare \otimes P) \vec{A}_0$,
      $G(P) \dfn \trans{\vec{B}_0}(\covare \otimes P) \vec{B}_0$,
      $H(P) \dfn \trans{\vec{B}_0}(\covare \otimes P) \vec{A}_0$.
    \item \label{the:nominal:third} The controller $K_{\infty}$ stabilizes \eqref{eq:sys1} in the mean square sense. 
    \item \label{the:nominal:second} The solution of the Riccati equation is found by solving the following \sdp{}: 
      \begin{equation} \label{eq:sdp}
        \begin{aligned}
          & \minimize_{P}  & & - \tr{P} \\
          & \stt & & \begin{bmatrix}
            Q - P + F(P) & \trans{H(P)} \\ H(P) & R + G(P)
          \end{bmatrix} \sgeq 0, \\
          &&& P \sgeq 0.\
        \end{aligned}
      \end{equation}
  \end{theoremcontents}

    
\end{proposition}
\begin{proof}
\archiv{See Appendix~\ref{app:nominal}.
}{
The proof of statement \ref{the:nominal:first} and \ref{the:nominal:third} follows that of \citep[Theorem 1]{Morozan1983} and the proof of \ref{the:nominal:second} follows from \cite{Balakrishnan2003} (details in \cite{Coppens2019}).
}
\end{proof}

Notice that the optimal solution to the LQR problem depends solely on the first and second moment of the random disturbance. This motivates our choice for the parametric form of the ambiguity set \eqref{eq:ambig2} used in the data-driven LQR problem, which we state in the next section.

\subsection{Data-driven stochastic LQR problem}
Consider now the case where $\probmul$ is not known a priori and only a finite set of offline \iid{} samples $\{\mulsdd\}_{i=0}^{M-1}$ is available. \dodd{For clarity, we add a hat to imply that a random variable depends on these samples.}
It is apparent that under such circumstances, it is only possible to solve \eqref{eq:lqr1} approximately. For most applications, however, it is crucial that the approximate solution remains stabilizing, which is not trivial. For instance, the \emph{empirical approach}, where \eqref{eq:lqr1} is solved using $\est{\mean} \dfn \frac{1}{\nsample} \ssum_{i=0}^{M-1} \mulsdd_i$ and $\est{\covar} \dfn \frac{1}{\nsample} \ssum_{i=0}^{M-1} (\mulsdd_i - \hat{\mu}) \trans{(\mulsdd_i - \hat{\mu})}$, does not guarantee stability. We will illustrate this with an example\archiv{}{\footnote{\label{fn:detailed}Detailed derivations in \cite{Coppens2019}.}}, which motivates the methodology presented in this paper.
\begin{example}[Motivating example]
   Consider the following scalar system,
\begin{equation*}
  \xs_{k+1} = (0.75 + \muls_k) \xs_k + \us_k,
\end{equation*}
where $\xs_k \in \Re$, $\us_k \in \Re$ and $\muls_k \in \Re$ are defined as before, but now $\muls$ is assumed Gaussian with $\Er{w}{\probmul} [\muls] = 0$ and $\Er{w}{\probmul} [\muls^2] = \sigma^2 = 0.5$. The empirical variance is $\est{\sigma}^2 = \tfrac{1}{\nsample} \ssum_{i=0}^{\nsample-1} \mulsdd_i^2$, using \iid{} samples $\{\mulsdd_i\}_{i=0}^{M-1}$. \archiv{We will develop an optimal LQR controller with a stage cost given by $q\xs_k^2 +  r \us_k^2 = \xs_k^2 +  10^4 \us_k^2$. Using the results from Proposition~\ref{the:nominal} to derive the Riccati equation, we obtain
\begin{equation*}
  q - p + (\est{\sigma}^2 + 0.75^2)p - (0.75 p)(r + p)^{-1}(0.75 p) = 0.
\end{equation*}
Note that this is a quadratic equation in $p$ with the following positive solution:
\begin{equation*}
  \dd{p}^* = \frac{\frac{q}{r} - 0.4375 + \est{\sigma}^2 + \sqrt{(\frac{q}{r}+  0.4375 - \est{\sigma}^2)^2 + 2.25 \frac{q}{r}}}{2(1-\est{\sigma}^2)}r \approx \frac{\est{\sigma}^2 - 0.4375}{1 - \est{\sigma}^2}r,
\end{equation*}
where we used $r \gg q$ and assumed that $\est{\sigma}^2 > 0.4375$ and $\est{\sigma}^2 < 1$. If the lower bound is not satisfied it turns out that the optimal feedback gain is given by $\dd{K}^* \approx 0$.
The optimal linear controller associated with this solution is given by:
\begin{equation} \label{eq:optimalg}
  \dd{K}^* = - \frac{0.75 \frac{\dd{p}^*}{r}}{1 + \frac{\dd{p}^*}{r}} \approx -\frac{\est{\sigma}^2 - 0.4375}{0.75}.
\end{equation} 
The closed-loop system given this controller is mean square stable \iff{} $(0.75 + \dd{K}^*)^2 + 0.5 ({\dd{K}^*})^2 < 1$, which follows from the recursive expression for the variance of the state in closed-loop: 
\begin{equation*}
  \Es{\muls}{\probmul}{\infty} \left[ (((0.75 + \dd{K}^*) + \dd{K}^* \muls) \xs_k)^2 \right] = ((0.75 + \dd{K}^*)^2 + \sigma^2 ({\dd{K}^*})^2) \Es{\muls}{\probmul}{\infty} \left[ \xs_k^2 \right].
\end{equation*}
Solving this stability condition for $\dd{K}^*$ leads to the conclusion that the system is mean square stable \iff{} $-1.4571 < \dd{K}^* < -0.0429$. Filling in \eqref{eq:optimalg} results in a condition on $\est{\sigma}^2$: $0.4697 < \est{\sigma}^2 < 1.5303$. Note that $\est{\sigma}^2 > 1$ implies that the system is not stabilizable (this is related to the uncertainty threshold principle of \cite{Athans1977}), so we only consider the lower bound on $\est{\sigma}^2$ (which is larger than $0.4375$ from our earlier assumption). In fact, we can now evaluate the probability that the empirical approach provides an unstable closed-loop controller as
\begin{equation*}
  \probmuls\left[\est{\sigma}^2 < 0.4697\right] = \probmuls\left[\ssum_{i=1}^{\nsample-1} \dd{\xi}_i^2 < \frac{0.4697\nsample}{\sigma^2} \right],
\end{equation*}
with $\dd{\xi}_i = \nicefrac{\mulsdd_i}{\sigma} \sim \mathcal{N}(0, 1)$ and which corresponds to the cumulative distribution function of a $\chi^2$-distributed random variable with $\nsample$ degrees of freedom, since we assumed that $w$ is Gaussian. Evaluating this probability numerically, we find that with $\nsample=500$, there is a probability of $0.1693$ that the empirical approach provides an unstable closed-loop controller.}{We develop an optimal LQR controller with a stage cost given by $q\xs_k^2 +  r \us_k^2 = \xs_k^2 +  10^4 \us_{k}^{2}$. Using results from Proposition~\ref{the:nominal}, one can derive that the optimal controller is approximately given by $\dd{K}^* \approx \nicefrac{(1 - 0.75^2 - \est{\sigma}^2)}{0.75}$ for $\est{\sigma}^2 > 0.4375$, otherwise $\dd{K}^* \approx 0$. We also assume that $\est{\sigma}^2 < 1$ (related to the uncertainty threshold principle of \cite{Athans1977}) since otherwise the problem is infeasible as the system is not stabilizable. We can show$^{\text{\ref{fn:detailed}}}$ that the controller $\dd{K}^*$ is \mss{} \iff{} $(0.75 + \dd{K}^*)^2 + 0.5 ({\dd{K}^*})^2 < 1$, which holds \iff{} $-1.4571 < \dd{K}^* < -0.0429$. This, in turn, is the case \iff{} $0.4697 < \est{\sigma}^2 < 1.5303$. Only the lower bound is critical since we assumed that $\est{\sigma}^2 < 1$. As $\est{\sigma}^2$ is a scaled $\chi^2$-distributed random variable with $\nsample$ degrees of freedom, the probability of $\est{\sigma}^2 < 0.4697$ occurring for $\nsample=500$ is $0.1693$. That is, the probability that the \emph{empirical approach} provides an unstable closed-loop controller in this example is almost 17\%.} 
\end{example}
From this example, it is clear that underestimation of the variance of the disturbance is directly related to the probability of failure of the controller. 
In order to take this into account, we introduce an arbitrarily-chosen confidence level $\beta \in [0, 1]$ and a corresponding ambiguity set $\amb: \Omega \rightrightarrows \mathcal{M}$, which represents the uncertainty of estimators $\est{\mean}$ and $\est{\covar}$. The size of $\mathcal{A}$ is then determined such that $\probmuls ( \probmul \in \amb ) \geq 1-\beta$.
In particular, we parametrize $\mathcal{A}$ as first suggested by \cite{Delage2010}:
\begin{equation} \label{eq:ambig2}
  \amb \dfn \left\{ \probmule \in \mathcal{M} \, \middle| \, \begin{array}{l}
    \trans{(\Ere{\mulse}{\probmule} \left[ \mulse \right] - \est{\mean})} \est{\covar}^{-1} (\Ere{\mulse}{\probmule} \left[ \mulse \right] - \est{\mean}) \leq \rmean \\
    \Ere{\mulse}{\probmule}\left[ (\mulse-\mean) \trans{(\mulse-\mean)} \right] \leq \rcovar \est{\covar},
  \end{array}\right\},
\end{equation}
The values of $\rmean$ and $\rcovar$ such that $\probmuls ( \probmul \in \amb ) \geq 1-\beta$ are derived in Section~\ref{sec:bounds}. In \Cref{sec:dr}, the following DR counterpart of \eqref{eq:lqr1} is solved
\begin{equation} \label{eq:lqr2}
  \begin{aligned}
    & \minimize_{\us_0, \us_1, \ldots} & & \underset{\probmule \in \amb}{\max}\, \Ese{\mulse}{\probmule}{\infty} \left[ \sum_{k=1}^{\infty} \trans{\xs_k} Q \xs_k + \trans{\us_k} R \us_k \right] \\
    & \stt & &\xs_{k+1} = A(\mulse_k)\xs_k+B(\mulse_k)\us_k, \quad k\in\N \\
    &&& \xs_0 = \xinits,
  \end{aligned}
\end{equation}
where $\{v_k\}_{k\in\N}$ is a trajectory of \iid{} copies of $v$. 
In doing so, we can finally establish \mss{} of the data-driven controller with high probability, by virtue of the following generalization of \Cref{the:lyapunov} to the DR case.

\begin{theorem}[DR Lyapunov stability] \label{the:drlyapunov}
  Consider the matrices $\{\hat{A}_i \}_{i=0}^\muld$, the set
  $\{w_k\}_{k\in\N}$ consisting of \iid{} copies of a square integrable random vector $w$ and the autonomous system $\xs_{k+1} = \dd{A}(\muls_k)\xs_k$, with $\dd{A}(\muls) = \ssum_{i=1}^{M} \dd{A}_i \muls^{(i)}$. Say that we have an ambiguity set with $\probmuls ( \probmul \in \amb ) \geq 1-\beta$, with $\probmul$ the true distribution of $w$. Then if there exists a $P \sgt 0$ such that:
  \begin{equation} \label{eq:drlyapunov}
    P - \max_{\probmule \in \amb} \Ere{\mulse}{\probmule} \left[ \trans{A(\mulse)}PA(\mulse) \right]\sgt 0,
  \end{equation} the autonomous system is \emss{} with probability at least $1-\beta$.
\end{theorem}
\begin{proof}
\archiv{See Appendix~\ref{app:drlyapunov}.
}{Full proof in \cite{Coppens2019}.}
\end{proof}



\section{Data-driven ambiguity set estimation} \label{sec:bounds}
We now turn to the problem of estimating the parameters $\rcovar$ and $\rmean$ involved in the definition of the ambiguity set \eqref{eq:ambig2}, given that we have $\nsample$ i.i.d. draws from the true distribution.
These parameters will be estimated under the following assumption on the disturbances.

\begin{definition}[Sub-Gaussianity\refinenv{Wainwright2019}] \label{def:subG}
  A random variable $y$ is sub-Gaussian with variance proxy $\sigma^2$ if $\Er{y}{\prob} [y] = 0$ and its moment generating function satisfies
  \begin{equation} \label{eq:subG}
    \Er{y}{\prob} [\exp(\lambda y)] \leq \exp\left( \frac{\sigma^2 \lambda^2}{2} \right) \quad \forall \lambda \in \Re. 
  \end{equation}
  We denote this by $y \sim \subG(\sigma^2)$. We say that a random vector $\xi \in \Re^{\muld}$ is sub-Gaussian, or $\xi \sim \subG_\muld(\sigma^2)$, if $\trans{z} \xi \sim \subG(\sigma^2), \forall z \in \Re^\muld$ with $\nrm{z}_2 = 1$. 
\end{definition}

\begin{assumption} \label{assump:subG}
We assume that \begin{inlinelist}
\item $\muls$ is square integrable; 
\item $\{\muls_k\}_{k\in \N}$ are \iid{} copies of $\muls$;
\item $\Sigma \succ 0$; and
\item \label{assump:subG:d} $\Sigma^{\nicefrac{-1}{2}}(\muls_k - \mu) \sim \subG_{\muld}(\sigma^2)$ for some $\sigma \geq 1$. 
\end{inlinelist}
\end{assumption}
Note that in the specific case of Gaussian disturbances, Assumption~\ref{assump:subG}\ref{assump:subG:d} holds with $\sigma^2=1$, so no further prior knowledge on the distribution is required. Moreover, in this case the bound on the covariance obtained in \Cref{thm:covariance-bound} can be slightly improved \citep{Wainwright2019}. In the case of bounded disturbances, $\sigma^2$ can be estimated in a data-driven fashion \citep{Delage2010}.

For the moment, we restrict our attention to obtaining concentration inequalities for moment estimators of random vectors with zero mean and unit variance --- hereafter referred to as \emph{isotropic} random vectors.
We will then convert these results into ambiguity sets of the form \eqref{eq:ambig2} using arguments from~\cite{Delage2010,So2011}. We begin by specializing the isotropic covariance bound, derived with constants in \cite[Lemma A.1.]{Hsu2012} based on a result by \cite{Litvak2005} and the isotropic mean bound by \cite[Theorem 2.1]{Hsu2012b}.
\begin{theorem}[Isotropic covariance bound\refinenv{Hsu2012}] \label{thm:covariance-bound}
  Let $\xi \sim \subG_{\muld}(\sigma^2)$ be a random vector, with $\Er{\xi}{\prob} [\xi] = 0$, $\Er{\xi}{\prob} [\xi \trans{\xi}] = I_\muld$. Let $\{\dd{\xi}_i\}_{i=0}^{\nsample-1}$ be $\nsample$ independent copies of $\xi$ and $\Ihat \dfn \tfrac{1}{\nsample} \sum_{i=0}^{\nsample-1} \dd{\xi}_i \trans{\dd{\xi}_i}$. Then 
  \begin{equation*}
    \prob[\|\Ihat - I_{\muld}\|_{2} \leq \tcovar(\beta)] \geq 1 - \beta, 
  \end{equation*}
    with 
    \begin{equation} \label{eq:def-tcovar}
      \tcovar(\beta) \dfn \frac{\sigma^2}{1 - 2 \epsilon} \left( \sqrt{\frac{32q(\beta, \epsilon, \muld)}{\nsample}} + \frac{2q(\beta, \epsilon, \muld)}{\nsample} \right),
    \end{equation}
  where $\epsilon \in \left( 0, \, \nicefrac{1}{2} \right)$ is chosen freely and $q(\beta, \epsilon, \muld) \dfn \muld \log{(1 + \nicefrac{1}{\epsilon})} + \log{(\nicefrac{2}{\beta})}$.
\end{theorem}
\begin{theorem}[Isotropic mean bound\refinenv{Hsu2012b}] \label{thm:mean-bound}
  Let $\{\dd{\xi}_i\}_{i=0}^{\nsample-1}$ be as defined in \Cref{thm:covariance-bound}, and 
  $\hat{\zeta} \dfn \tfrac{1}{\nsample} \sum_{i=0}^{\nsample-1} \dd{\xi}_i$. Then
  \begin{equation*}
    \prob[\nrm{\hat{\zeta}}_2^2 \leq \tmean(\beta)] \geq 1 - \beta.
  \end{equation*}
  where 
  \begin{equation*}
    \tmean(\beta) \dfn \frac{\sigma^2}{\nsample} p(\beta, \muld),
  \end{equation*}
  with $p(\beta, \muld) \dfn \left(\muld + 2\sqrt{\muld\log{(\nicefrac{1}{\beta})}} + 2 \log{(\nicefrac{1}{\beta})}\right)$. 
\end{theorem} 
By combining the bounds in \Cref{thm:covariance-bound,thm:mean-bound}, we readily obtain the following result.

\begin{theorem}[Ambiguity set] \label{thm:final-bounds}
Let $\muls \in \Re^{\muld}$ be a sub-Gaussian random vector, with $\Er{\muls}{\prob}[\muls] = \mean$, $\Er{\muls}{\prob}[(\muls-\mean)\trans{(\muls-\mean)}] = \covar$ and $\xi = \covar^{-\nicefrac{1}{2}} (\muls - \mean) \sim \subG{}_\muld(\sigma^2)$. Let $\{ \dd{\muls}_i\}_{i=0}^{\nsample-1}$ be independent copies of $\muls$. Let $\hat{\mean} \dfn \tfrac{1}{\nsample}\sum_{i=0}^{\nsample-1} \dd{\muls}_i$ and $\hat{\covar} \dfn \tfrac{1}{\nsample}\sum_{i=0}^{\nsample-1} (\dd{\muls}_i - \hat{\mean}) \trans{(\dd{\muls}_i-\hat{\mean})}$ denote the empirical estimators for the mean and 
the covariance matrix, respectively. 
Let $\epsilon$, $p(\beta, \muld)$, $q(\beta, \epsilon, \muld)$, $\tcovar(\nicefrac{\beta}{2})$ and $\tmean(\nicefrac{\beta}{2})$ be as defined in \Cref{thm:covariance-bound,thm:mean-bound}. Provided that
\begin{equation} \label{eq:condition-sample-size}
  \nsample > \left(\tfrac{\sigma^2 \sqrt{32q(\nicefrac{\beta}{2}, \epsilon, \muld)} + \sqrt{32 \sigma^4q(\nicefrac{\beta}{2}, \epsilon, \muld) + 8\sigma^2(1-2\epsilon)q(\nicefrac{\beta}{2}, \epsilon, \muld) + 4\sigma^2 (1-2\epsilon)^2 p(\nicefrac{\beta}{2}, \muld) }}{2(1 - 2\epsilon)}\right)^2,
\end{equation}
then with probability at least $1-\beta$,
\[ 
  \begin{aligned}
  \trans{(\hat{\mean} - \mean)}\hat{\covar}^{-1}(\hat{\mean} - \mean) \leq \rmean, \quad
  \covar \preceq \rcovar \hat{\covar}, 
  \end{aligned}
\]
with $\rcovar \dfn \tfrac{1}{1- \tmean(\nicefrac{\beta}{2}) - \tcovar(\nicefrac{\beta}{2})}$ and $\rmean \dfn \tfrac{\tmean(\nicefrac{\beta}{2})}{1 - \tmean(\nicefrac{\beta}{2}) - \tcovar(\nicefrac{\beta}{2})}$.

\end{theorem}

\begin{proof}
  \archiv{
  Let us define $\xi = \covarmhalf (\muls- \mean) \sim \subG_{\muld}(\sigma^2)$ so that 
  \[ 
      \begin{aligned}
          \Er{\muls}{\prob}[\xi] &= \covarmhalf (\Er{\muls}{\prob}[\muls] - \mean) = 0 \\
          \Er{\muls}{\prob}[\xi\trans{\xi}] &= \covarmhalf \Er{\muls}{\prob}\left[(\muls - \mean) \trans{(\muls - \mean)}\right] \covarmhalf = \covarmhalf \covar \covarmhalf = I_\muld.
      \end{aligned}  
  \]
  Let $\hat{\zeta} = \tfrac{1}{\nsample} \sum_{i=0}^{\nsample-1} \dd{\xi}_i$ and $\Ihat = \tfrac{1}{\nsample} \sum_{i=0}^{\nsample-1} \dd{\xi}_i \trans{\dd{\xi}_i}$.
  By \Cref{thm:covariance-bound,thm:mean-bound}, we then have with 
  probability $1-\beta$ that,
  \begin{subequations} 
    \begin{align} \label{eq:isotropic-covar}
      \|\Ihat - I_\muld \|_2 &\leq \tcovar(\nicefrac{\beta}{2}) \\ 
      \text{and} \; \label{eq:isotropic-mean}
      \|\hat{\zeta}\|_2^2 &\leq \tmean(\nicefrac{\beta}{2}).
    \end{align}
  \end{subequations}
  Let us define the covariance estimator with respect to the true mean by
  \begin{equation} \label{eq:proof-isotropic-covar}
    \tilde{\covar} \dfn \tfrac{1}{\nsample} \ssum_{i=0}^{\nsample-1} (\dd{\muls}_i - \mean) \trans{(\dd{\muls}_i - \mean)} = \tfrac{1}{\nsample} \ssum_{i=0}^{\nsample-1} (\covarhalf \xi_i) \trans{(\covarhalf \xi_i)} = \covarhalf \Ihat \covarhalf.
  \end{equation}
   Therefore, we can rewrite \eqref{eq:isotropic-covar} in terms of the anisotropic random variable $\muls$ as 
  \begin{equation} \label{eq:proof-bound-covar}
    \begin{aligned}  
      \|\Ihat - I_\muld \|_2 \leq \tcovar(\nicefrac{\beta}{2}) &\Leftrightarrow (1-\tcovar(\nicefrac{\beta}{2})) I_\muld \preceq \Ihat \preceq (1+\tcovar(\nicefrac{\beta}{2})) I_\muld \\ 
      &\Rightarrow (1-\tcovar(\nicefrac{\beta}{2})) \covar \preceq \tilde{\covar} \preceq (1+\tcovar(\nicefrac{\beta}{2}))\covar, 
    \end{aligned}
    \end{equation}
    where we have used \eqref{eq:proof-isotropic-covar} and the fact that condition \eqref{eq:condition-sample-size} implies $\tcovar(\nicefrac{\beta}{2}) < 1$.
  As shown in \cite[Thm. 2]{Delage2010}, \eqref{eq:isotropic-mean} implies that for any $x \in \Re^\muld$,
  \begin{equation} \label{eq:proof-bound-mean}
  \begin{aligned}
    \trans{x} (\hat{\mean} -\mean)\trans{(\hat{\mean} - \mean)} x &= (\trans{x} (\hat{\mean} - \mean))^2 = (\trans{x} \smashoverbracket{\covarhalf \covarmhalf}{I_\muld} (\hat{\mean}-\mean))^2 \\
    &\leq \| \covarhalf x \|_2^2 \| \covarmhalf (\hat{\mean}-\mean) \|_2^2 \\
    &\leq \tmean(\nicefrac{\beta}{2}) \trans{x} \covar x.
  \end{aligned}
\end{equation}
  Using this fact, we can bound $\tilde{\covar}$ with respect to $\hat{\covar}$ and $\covar$ as 
  \[
      \begin{aligned}
      \tilde{\covar} &= \tfrac{1}{\nsample} \ssum_{i=0}^{\nsample-1} (\muls_i - \mean)\trans{(\muls_i - \mean)} \\ 
      &=\tfrac{1}{\nsample} \ssum_{i=0}^{\nsample-1} (\muls_i - \hat{\mean} + \hat{\mean} - \mean)\trans{(\muls_i - \hat{\mean} + \hat{\mean} - \mean)}\\
      &=\tfrac{1}{\nsample} \ssum_{i=0}^{\nsample-1} (\muls_i - \hat{\mean})\trans{(\muls_i - \hat{\mean})}
      + (\muls_i - \hat{\mean})\trans{(\hat{\mean} - \mean)}
      + (\hat{\mean} - \mean) \trans{(\muls_i - \hat{\mean})}
      + (\hat{\mean} - \mean)\trans{(\hat{\mean} - \mean)}\\
      &= \hat{\covar} + (\hat{\mean} - \mean) \trans{(\hat{\mean} - \mean)} \preceq \hat{\covar} + \tmean(\nicefrac{\beta}{2}) \covar,
      \end{aligned}
  \]
  where we used \eqref{eq:proof-bound-mean} in the final step. Combining this with \eqref{eq:proof-bound-covar}, we have that $(1-\tcovar)\covar \preceq \hat{\covar} + \tmean\covar$, thus
  \[ 
      \covar \leq \tfrac{1}{1- \tmean(\nicefrac{\beta}{2}) - \tcovar(\nicefrac{\beta}{2})} \hat{\covar},
  \]
  \[
    (1 - \tmean - \tcovar)\trans{(\hat{\mean} - \mean)} \hat{\covar}^{-1} (\hat{\mean}- \mean) \leq \trans{(\hat{\mean} - \mean)} \covar^{-1} (\hat{\mean} - \mean) = \trans{\xi} \xi \leq \tmean.
  \] 
  Condition~\ref{eq:condition-sample-size} then follows from assuming $1 - \tmean - \tcovar > 0$, which is a quadratic inequality in $\sqrt{\nsample}$.}{Apply the procedure of \cite[Thm. 2]{Delage2010} to the results of \Cref{thm:covariance-bound}--\ref{thm:mean-bound}.}
\end{proof} 


\section{Distributionally Robust LQR} \label{sec:dr} 
We will tackle the solution of \eqref{eq:lqr1} for the ambiguity set given in \eqref{eq:ambig2} in two stages. Firstly, we extend the result of Proposition~\ref{the:nominal} to the case where $\mean$ is known and $\covar$ is estimated, \ie{} $\rmean = 0$ and $\rcovar > 0$. Secondly, we present the result where both the mean and the covariance are estimated.

\subsection{Uncertain covariance}  \label{sec:covar}
The case where the mean is known is interesting since we can still formulate an exact solution to \eqref{eq:lqr2}. This will no longer be true for the full-uncertainty case (\Cref{sec:mean}).

  \begin{proposition} \label{the:ukvar}
    Consider that $\mulse \in \Re^\muld$ is distributed according to an element of the set 
    \begin{equation} \label{eq:ambigcovar}
      \amb_\covar \dfn \left\{ \probmule \in \mathcal{M} \, \mid \Ere{\mulse}{\probmule} \left[ (\mulse-\mean) \trans{(\mulse-\mean)} \right] \sleq \rcovar \hat{\covar}, \, \Ere{\mulse}{\probmule} \left[ \mulse \right] = \mean \right\},
    \end{equation}
    where $\probmuls( \probmul \in \amb_\covar ) \geq 1-\beta$. Then applying Proposition~\ref{the:nominal} with $\covar = \rcovar \est{\covar}$ results in the optimal linear controller for \eqref{eq:lqr2}, assuming that \eqref{eq:sys1} is DR mean square stabilizable, \ie{} there exists a $K$ such that the DR Lyapunov decrease \eqref{eq:drlyapunov} holds for the closed-loop system $x_{k+1} = (A(\muls_k) + B(\muls_k) K)\xs_k$. The optimal controller is also mean square stabilizing for \eqref{eq:sys1} with probability at least $1 - \beta$. 
  \end{proposition}
  \begin{proof}
    \archiv{See Appendix~\ref{app:ukvar}.
    }{The Bellman operator associated with \eqref{eq:lqr2} is of the same form as the one in the proof of Proposition~\ref{the:nominal}, which makes it applicable (full proof is in \citep{Coppens2019}).}
  \end{proof}

\subsection{Full uncertainty} \label{sec:mean}
We finally consider the more general case using the full ambiguity set $\amb$ given by \eqref{eq:ambig2}. The general \emph{min-max problem} \eqref{eq:lqr2} for such sets is computationally intractable, which is why an upper bound on the quadratic cost is minimized instead by employing results from robust control \citep{Boyd1994, Kothare1996}. A common approach is to assume that the value function can be written in the quadratic form $V(\xs) = \trans{\xs}P\xs$ for some $P \sgt 0$ and to solve the optimization problem
\begin{equation} \label{eq:quadraticbound}
  \begin{aligned}
    & \minimize_{V(\xs)} & & \Er{\xinitsr}{\probxinit} \left[V(\xinitsr)\right] \\
    & \stt & & \Val(\xs) \geq \min_{\us} \left\{ \trans{\xs}Q\xs + \trans{\us}R\us + \max_{\probmule \in \amb} \Ere{\mulse}{\probmule} \left[ \Val(A(\mulse)x + B(\mulse)u) \right]\right\}, \quad \forall \xs,
  \end{aligned}
\end{equation}
where we introduced the random initial state $\xinitsr \in \Re^\xd$. The optimal cost of \eqref{eq:quadraticbound} then upper bounds the true LQR cost of \eqref{eq:lqr2} for a given value of $\xinits$ as proven in \cite{Kothare1996} for a similar setup. We can then write \eqref{eq:quadraticbound} as an \sdp{} using the following theorem.
\begin{theorem} \label{the:ukmean}
Let $\amb$ be an ambiguity set of the form \eqref{eq:ambig2}. Then we can find an approximate solution of \eqref{eq:quadraticbound} for the system \eqref{eq:sys1}, assuming that the initial state is given by the random vector $\xinitsr \in \Re^\xd$ with $\Er{\xinitsr}{\probxinit} \left[ \xinitsr \right] = 0$ and $\Er{\xinitsr}{\probxinit} \left[ \xinitsr \trans{\xinitsr} \right] = I_{\xd}$, by solving the following \sdp.
  \begin{subequations} \label{eq:lmiukm2}
      \begin{alignat}{4}
      & \maximize_{W, V, S, L} \quad & & \tr{W} \nonumber \\
      & \stt & & \smallmat{
          S & \rmean \trans{H_1} & \rmean \trans{H_2} & \cdots & \rmean \trans{H_\muld} \\
          \rmean H_1 & L &  &  &  \\
          \rmean H_2 &  & L &  & \\
          \vdots &  &  & \ddots & \\
          \rmean H_\muld &  &  &  & L
        } \sgeq 0, \label{seq:mlmi2} \\
      & & & \smallmat{
        W - \sqrt{2} S & \trans{(\vec{A}W + \vec{B}V)} & \trans{(\hat{A}W + \hat{B}V)}  & \trans{W}Q^{\frac{1}{2}} & \trans{V} R^{\frac{1}{2}} \\
        \vec{A}W + \vec{B}V & \covardr^{-1} \otimes W & & & \\
        \hat{A}W + \hat{B}V & & W - \sqrt{2} L & & \\
        Q^{\frac{1}{2}}W & & & I_{\xd} & \\
        R^{\frac{1}{2}}V & & & & I_{\ud}
      } \sgeq 0 \label{seq:mlmi3},
    \end{alignat}
\end{subequations}
where $H_i = \ssum_{j=1}^\muld [\hat{\covar}^{\nicefrac{1}{2}}]_{ji} (A_j W + B_j V)$, $\hat{A} = A(\hat{\mean})$, $\hat{B} = B(\hat{\mean})$, $\covardr = \rcovar \hat{\covar}$. Let $\dd{W}$ and $\dd{V}$ denote the minimizers of \eqref{eq:lmiukm2}. The corresponding linear controller $\us = \dd{K} \xs$ with $\dd{K} = \dd{V}\dd{W}^{-1}$ then achieves an upper bound of the cost \eqref{eq:quadraticbound}, given by $\Er{\xinitsr}{\probxinit} [ \trans{\xinitsr}\dd{P}\xinitsr]$, where $\dd{P} = \dd{W}^{-1}$. Moreover, $\dd{K}$ is mean-square stabilizing for \eqref{eq:sys1} with probability at least $1 - \beta$.
\end{theorem}
\begin{proof}
\archiv{
  See Appendix~\ref{app:ukmean}.
  } {The proof specializes \cite[Theorem 6.2.1]{Ben-Tal2000} (see \cite{Coppens2019}).}
\end{proof}

\archiv{
\begin{remark} \label{rem:trace}
  Two approximations are made in Theorem~\ref{the:ukmean}. First, leveraging \cite[Theorem 6.2.1]{Ben-Tal2000}, introducing an approximation error quantified by an increase of $\rmean$ by a factor of at most $\sqrt{\muld}$. Secondly, we minimize an upper-bound of the LQR cost instead of the cost itself. We can further decrease the closed-loop cost by instead using a receding horizon controller for a given $\xinits$, which too can be formulated as an SDP. Then \eqref{eq:quadraticbound} is reformulated as \citep{Kothare1996}:
  \begin{subequations}
    \begin{alignat}{4}
    & \minimize_{W, V, S, L, \gamma} \quad & & \gamma \nonumber \\
    & \stt & & \trans{\xinits} W^{-1} \xinits \leq \gamma \label{seq:mlmi4}  \\
    & & & \eqref{seq:mlmi2}, \eqref{seq:mlmi3}, \nonumber
  \end{alignat}
  \end{subequations}
  where \eqref{seq:mlmi4} can be replaced by a linear matrix inequality  using Schur's complement \citep[Sec. 2.1]{Boyd1994}. The assumption used in Theorem~\ref{the:ukmean} ensures that the solution converges to the optimal one as $\rmean$ and $\rcovar$ go to zero \citep{Balakrishnan2003}.
\end{remark}
}{
  \begin{remark} \label{rem:trace}
    Two approximations are made in Theorem~\ref{the:ukmean}. First, \cite[Theorem 6.2.1]{Ben-Tal2000}, loosens $\rmean$ by a factor of at most $\sqrt{\muld}$. Secondly, we minimize an upper-bound of the LQR cost instead of the cost itself. We can further decrease the closed-loop cost by using a receding horizon controller for a given $\xinits$, which too can be formulated as an SDP \citep{Kothare1996, Coppens2019}. By contrast, the assumption in Theorem~\ref{the:ukmean} ensures that the solution converges to the nominal one as $\rmean$ and $\rcovar$ go to zero \citep{Balakrishnan2003}.
  \end{remark}  
}
\begin{remark}
  Invertibility of $\rcovar \hat{\covar}$ can be guaranteed by using $\rcovar \hat{\covar} + \lambda I$ instead of $\rcovar \hat{\covar}$ for some small $\lambda$, which increasing the size of the ambiguity set, introducing additional conservatism.
\end{remark}
\section{Numerical Experiment} \label{sec:numerical}
We experimentally quantify the sample complexity of our approach, \ie{} how many samples are needed before the controller becomes equivalent to the nominal one based on the true $\covar$ and $\mean$ instead of their data-driven estimates. Consider the double integrator model with matrices:
\begin{equation*}
  A_0 = \mat{1 & T_s \\ 0 & 1-0.4T_s}, \, B_0 = \mat{0 \\ T_s}, \, A_1 = \mat{0 & 0 \\ 0 & -T_s}, \, A_2 = \mat{0 & 0 \\ 0 & 0}, \, B_1 = \mat{0 \\ 0}, \, B_2 = \mat{0 \\ T_s},
\end{equation*}
where we chose $T_s = 0.02$. The dynamics are then given by \eqref{eq:sys1} with $\muls_k$ an independent random sequence of Gaussian random vectors with covariance $\covar = \smallmat{1 & 0 \\ 0 & 1}$ and mean $\mean = \smallmat{0 \\ 0}$. 

To estimate the sample complexity, we determine controllers satisfying \eqref{eq:lqr2} for $Q = \smallmat{10 & 0 \\ 0 & 1}$ and $R = 0.01$. The parameters of the ambiguity set \eqref{eq:ambig2} are determined using Theorem~\ref{thm:final-bounds} with $\beta = 0.05$ and $\epsilon = \nicefrac{1}{30}$. Since $\muls_k$ are Gaussian, $\xi_k = \covar^{-\nicefrac{1}{2}} (w_k - \mu)$ are sub-Gaussian with $\sigma^2 = 1$. 

The simulation setup is as follows. We compare the nominal controller, the uncertain covariance controller (Proposition~\ref{the:ukvar}) and the full uncertainty controller (Theorem~\ref{the:ukmean}). We evaluate the expected closed-loop cost for $\xinits = \trans{\smallmat{2& 2}}$ by solving the Lyapunov equation. We start with $\nsample = 1000$ to satisfy \eqref{eq:condition-sample-size}. For each value of $\nsample$ we produce $30$ realisations of both DR controllers. \Cref{fig:sample-complexity} depicts confidence intervals for relative difference between closed-loop cost of the DR controllers and the nominal controller (\ie{} the relative suboptimality). The figure shows that both converge with a rate $\mathcal{O}(1/\nsample)$ to the nominal cost even though \Cref{the:ukmean} only solves \eqref{eq:lqr2} approximately.

\begin{figure}[t]
  \begin{minipage}[t]{0.5\textwidth}
    \vspace*{0pt}
    \input{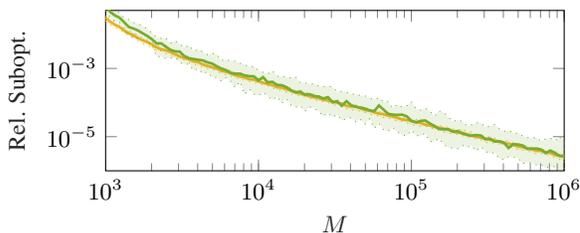}
  \end{minipage}%
  \hfill
  \begin{minipage}[t]{0.45\textwidth}
    \vspace*{0pt}
    \captionsetup{indention=0pt, labelfont=bf}
    \caption{
      Relative suboptimality versus sample size for controllers. The colored area depicts the $0.3$-confidence interval around the cost. The full line depicts the mean.} \label{fig:sample-complexity}
  \end{minipage}
  \hfill
  \vspace{-0.3cm}
  \end{figure}

\section{Conclusion and future work}
We studied the infinite horizon LQR problem for systems with multiplicative uncertainty on both the states and the inputs.
We operate in the setting where the distributions are estimated from data.
We show that using results from high-dimensional statistics, high-confidence ambiguity sets can be constructed, which allow us to formulate a DR counterpart to the stochastic optimal control problem as an SDP.
As a result, stability of the closed-loop system can be guaranteed with high probability.

In future work, we aim to perform an in-depth analysis of the conservatism introduced by the proposed formulations. Furthermore, we plan to study extensions towards DR Kalman filtering. 



\acks{This work was supported by the Ford KU Leuven Research Alliance; the Fonds Wetenschappelijk Onderzoek PhD grant 11E5520N and research projects G0A0920N, G086518N and G086318N; Research Council KU Leuven C1 project No. C14/18/068; Fonds de la Recherche Scientifique – FNRS and the Fonds Wetenschappelijk Onderzoek – Vlaanderen under EOS project no 30468160 (SeLMA).}

\bibliography{bibliography.bib}

\begin{thebibliography}{22}
\providecommand{\natexlab}[1]{#1}
\providecommand{\url}[1]{\texttt{#1}}
\expandafter\ifx\csname urlstyle\endcsname\relax
  \providecommand{\doi}[1]{doi: #1}\else
  \providecommand{\doi}{doi: \begingroup \urlstyle{rm}\Url}\fi

\bibitem[{Athans} et~al.(1977){Athans}, {Ku}, and {Gershwin}]{Athans1977}
Michael {Athans}, Richard {Ku}, and Stanley~B. {Gershwin}.
\newblock The uncertainty threshold principle: Some fundamental limitations of
  optimal decision making under dynamic uncertainty.
\newblock \emph{IEEE Transactions on Automatic Control}, 22\penalty0
  (3):\penalty0 491--495, June 1977.

\bibitem[Balakrishnan and Vandenberghe(2003)]{Balakrishnan2003}
Venkataramanan Balakrishnan and Lieven Vandenberghe.
\newblock {Semidefinite programming duality and linear time-invariant systems}.
\newblock \emph{IEEE Transactions on Automatic Control}, 48\penalty0
  (1):\penalty0 30--41, 2003.

\bibitem[Ben-Tal et~al.(2000)Ben-Tal, El~Ghaoui, and Nemirovski]{Ben-Tal2000}
Aharon Ben-Tal, Laurent El~Ghaoui, and Arkadi Nemirovski.
\newblock Robustness.
\newblock In \emph{Handbook of Semidefinite Programming}, pages 139--162.
  Springer US, Boston, MA, 2000.

\bibitem[Billingsley(1995)]{Billingsley1995}
Patrick Billingsley.
\newblock \emph{Probability and measure}.
\newblock Wiley series in probability and mathematical statistics. Wiley, New
  York, 3rd edition edition, 1995.

\bibitem[Boyd et~al.(1994)Boyd, {El Ghaoui}, Feron, and Balakrishnan]{Boyd1994}
Stephen Boyd, Laurent {El Ghaoui}, Eric Feron, and Venkataramanan Balakrishnan.
\newblock \emph{{Linear Matrix Inequalities in System and Control Theory}}.
\newblock Society for Industrial and Applied Mathematics, 1994.

\bibitem[Byrnes(1979)]{Byrnes1979}
Christopher~I. Byrnes.
\newblock {On the stabilizability of linear control systems depending on
  parameters}.
\newblock In \emph{Proceedings of the IEEE Conference on Decision and Control},
  volume~1, pages 233--236. IEEE, 1979.

\bibitem[Costa and Kubrusly(1997)]{Costa1997}
Oswaldo~L.V. Costa and Carlos~S. Kubrusly.
\newblock {Quadratic optimal control for discrete-time infinite-dimensional
  stochastic bilinear systems}.
\newblock \emph{IMA Journal of Mathematical Control and Information},
  14\penalty0 (4):\penalty0 385--399, 1997.

\bibitem[Dean et~al.(2019)Dean, Mania, Matni, Recht, and Tu]{Dean2019}
Sarah Dean, Horia Mania, Nikolai Matni, Benjamin Recht, and Stephen Tu.
\newblock {On the Sample Complexity of the Linear Quadratic Regulator}.
\newblock \emph{Foundations of Computational Mathematics}, 2019.

\bibitem[Delage and Ye(2010)]{Delage2010}
Erick Delage and Yinyu Ye.
\newblock {Distributionally robust optimization under moment uncertainty with
  application to data-driven problems}.
\newblock \emph{Operations Research}, 58\penalty0 (3):\penalty0 595--612, 2010.

\bibitem[Dupačová(1987)]{dupacova_minimax_1987}
Jitka Dupačová.
\newblock The minimax approach to stochastic programming and an illustrative
  application.
\newblock \emph{Stochastics}, 20\penalty0 (1):\penalty0 73--88, 1987.

\bibitem[{El Ghaoui}(1995)]{ElGhaoui1995}
Laurent {El Ghaoui}.
\newblock {State-feedback control of systems with multiplicative noise via
  linear matrix inequalities}.
\newblock \emph{Systems and Control Letters}, 24\penalty0 (3):\penalty0
  223--228, 1995.

\bibitem[Gravell et~al.(2019)Gravell, Esfahani, and Summers]{Gravell2019}
Benjamin Gravell, Peyman~Mohajerin Esfahani, and Tyler Summers.
\newblock Learning robust control for {LQR} systems with multiplicative noise
  via policy gradient.
\newblock \emph{arXiv preprint arXiv:1905.13547}, 2019.

\bibitem[Hsu et~al.(2012{\natexlab{a}})Hsu, Kakade, and Zhang]{Hsu2012}
Daniel Hsu, Sham~M. Kakade, and Tong Zhang.
\newblock Tail inequalities for sums of random matrices that depend on the
  intrinsic dimension.
\newblock \emph{Electronic Communications in Probability}, 17\penalty0
  (52):\penalty0 1--13, 2012{\natexlab{a}}.

\bibitem[Hsu et~al.(2012{\natexlab{b}})Hsu, Kakade, and Zhang]{Hsu2012b}
Daniel Hsu, Sham~M. Kakade, and Tong Zhang.
\newblock A tail inequality for quadratic forms of subgaussian random vectors.
\newblock \emph{Electronic Communications in Probability}, 17\penalty0
  (52):\penalty0 1--6, 2012{\natexlab{b}}.

\bibitem[Kothare et~al.(1996)Kothare, Balakrishnan, and Morari]{Kothare1996}
Mayuresh~V. Kothare, Venkataramanan Balakrishnan, and Manfred Morari.
\newblock {Robust constrained model predictive control using linear matrix
  inequalities}.
\newblock \emph{Automatica}, 32\penalty0 (10):\penalty0 1361--1379, 1996.

\bibitem[Litvak et~al.(2005)Litvak, Pajor, Rudelson, and
  Tomczak-Jaegermann]{Litvak2005}
Alexander~E. Litvak, Alalin Pajor, M.~Rudelson, and N.~Tomczak-Jaegermann.
\newblock {Smallest singular value of random matrices and geometry of random
  polytopes}.
\newblock \emph{Advances in Mathematics}, 195\penalty0 (2):\penalty0 491--523,
  2005.

\bibitem[Morozan(1983)]{Morozan1983}
Toader Morozan.
\newblock {Stabilization of some stochastic discrete-time control systems}.
\newblock \emph{Stochastic Analysis and Applications}, 1\penalty0 (1):\penalty0
  89--116, 1983.

\bibitem[Schuurmans et~al.(2019)Schuurmans, Sopasakis, and
  Patrinos]{Schuurmans2019safe}
Mathijs Schuurmans, Pantelis Sopasakis, and Panagiotis Patrinos.
\newblock Safe learning-based control of stochastic jump linear systems: a
  distributionally robust approach.
\newblock In \emph{2019 IEEE Conference on Decision and Control (CDC)}, pages
  6498--6503. IEEE, 2019.

\bibitem[So(2011)]{So2011}
Anthony Man~Cho So.
\newblock {Moment inequalities for sums of random matrices and their
  applications in optimization}.
\newblock \emph{Mathematical Programming}, 130\penalty0 (1):\penalty0 125--151,
  2011.

\bibitem[Wainwright(2019)]{Wainwright2019}
Martin~J. Wainwright.
\newblock \emph{{High-Dimensional Statistics}}.
\newblock Cambridge University Press, 2019.

\bibitem[Wonham(1967)]{Wonham1967}
W.~Murray Wonham.
\newblock {Optimal Stationary Control of a Linear System with State-Dependent
  Noise}.
\newblock \emph{SIAM Journal on Control}, 5\penalty0 (3):\penalty0 486--500,
  1967.

\bibitem[Wu et~al.(1996)Wu, Yang, Packard, and Becker]{Wu1996}
Fen Wu, Xin~Hua Yang, Andy Packard, and Greg Becker.
\newblock {Induced L2-norm control for LPV systems with bounded parameter
  variation rates}.
\newblock \emph{International Journal of Robust and Nonlinear Control},
  6\penalty0 (9-10):\penalty0 983--998, 1996.

\end{thebibliography}

\archiv{
\appendix

\section{Proofs for nominal case}
\subsection{Proof of Theorem~\ref{the:lyapunov}} \label{app:lyapunov}
  \citep[Lemma 1]{Morozan1983} states that the following is a necessary and sufficient condition for \mss:
  \begin{equation*}
    P - \Er{\muls}{\probmul} \left[ \trans{A(\muls)} P A(\muls) \right] \sgt 0.
  \end{equation*}
  The expected value is of the form $\Er{\muls}{\probmul} \left[ \trans{v} P v \right]$, with $v$ a random vector. We will leverage the following, easily verified result:
  \begin{equation} \label{eq:quadexp}
    \Er{\muls}{\probmul} \left[ \trans{v} P v \right] = \tr{\left(P \Er{\muls}{\probmul} \left[ \tilde{v}\trans{\tilde{v}} \right]\right)} + \trans{\Er{\muls}{\probmul} \left[ v \right]} P \Er{\muls}{\probmul} \left[ v \right],
  \end{equation}
  with $\tilde{v} = v - \Er{\muls}{\probmul} \left[ v \right]$. Taking $v = A(\muls)\xs$ for a given $\xs$ results in:
  \begin{equation*}
    \begin{split}
      \Er{\muls}{\probmul} \left[ v \right] &= A(\mean)\xs, \\
      \Er{\muls}{\probmul} \left[ \tilde{v} \trans{\tilde{v}} \right] &= \Er{\muls}{\probmul} \left[ (\trans{\tilde{\muls}} \otimes I_\xd) \vec{A} \xs \trans{\xs} \trans{\vec{A}} (\tilde{\muls} \otimes I_\xd) \right] = \ssum_{i, j=1}^\muld \left( \covar_{ij} A_i \xs \trans{\xs} \trans{A_j} \right),
    \end{split}
  \end{equation*}
  with $\tilde{\muls} = \muls - \mean$. Applying \eqref{eq:quadexp} then gives:
  \begin{equation*}
    \begin{split}
      \E_{\probmul} \left[ \trans{\xs} \trans{A(\muls)} P A(\muls) \xs \right] &= \tr{\left(P \ssum_{i, j=1}^\muld \left( \covar_{ij} A_i \xs \trans{\xs} \trans{A_j} \right) \right)} + \trans{\xs} \trans{A(\mean)} P A(\mean) \xs \\
      &= \tr{\left( \textstyle \ssum_{i, j=1}^\muld\left( \covar_{ij} \trans{\xs} \trans{A_j}  P  A_i \xs  \right) \right)} + \trans{\xs} \trans{A(\mean)} P A(\mean) \xs \\
      &= \trans{\xs} \trans{\vec{A}} (\covar \otimes P) \vec{A} \xs + \trans{\xs} \trans{A(\mean)} P A(\mean) \xs \\
      &= \trans{\xs} \trans{\vec{A}_0} (\covare \otimes P) \vec{A}_0 \xs.
    \end{split}
  \end{equation*}
  By \cite[Lemma 1]{Morozan1983} it follows that \mss{} is equivalent to \emss{} for system \eqref{eq:sys1}. \hfill \QED

\subsection{Proof of Proposition~\ref{the:nominal}} \label{app:nominal}
  We will follow the proof of \citep[Theorem 1]{Morozan1983}. We look at properties of the Bellman operator associated with \eqref{eq:lqr1}:
  \begin{equation} \label{eq:bellman}
      (T \Val_k)(\xs) = \min_{\us} \left\{ \trans{\xs}Q\xs + \trans{\us}R\us + \Er{\muls}{\probmul} \left[ \Val_k(A(\muls)\xs + B(\muls)\us) \right]\right\},
  \end{equation}
  Assuming a quadratic value function of the form $V_k(x) = V(x) = \trans{\xs}P\xs, \forall k \in \N$, we can for all $\xs \in \Re^\xd$, write the $k+1$'th step of value iteration with $P_0 = 0$ as
  \begin{subequations}
    \begin{alignat}{3}
      \trans{x} P_{k+1} x &= \min_{\us} \left\{\trans{\xs} Q \xs + \trans{\us}R\us + \trans{(\vec{A}_0 \xs + \vec{B}_0\us)} (\covare \otimes P_k) (\vec{A}_0 \xs + \vec{B}_0\us) \right\} \nonumber \\
      &= \trans{\xs} (Q + \trans{K_k}RK_k) \xs + \trans{\xs}\trans{(\vec{A}_0 + \vec{B}_0 K_k)} (\covare \otimes P_k) (\vec{A}_0 + \vec{B}_0 K_k)\xs  \label{eq:bellman3} \\
      \Leftrightarrow P_{k+1} &= Q + F(P_k) - \trans{H(P_k)}(R + G(P_k))^{-1}H(P_k) \label{eq:bellman4},
    \end{alignat}
  \end{subequations}
  where the first equality follows from the same arguments as those in the proof of Theorem~\ref{the:lyapunov}. The second equality follows from the fact that the optimal value of $\us$ in \eqref{eq:bellman} is given by $\us = K_k \xs$ with $K_k = -(R + G(P_k))^{-1}H(P_k)$.

  Statement \ref{the:nominal:first} follows from the following statements, which we will prove separately:
  \begin{inlinelist}
    \item value iteration converges to some $P_\infty$; \label{eq:nominal:first:sub1}
    \item Given an initial state $\xs_0$, $\trans{\xs_0} P_\infty \xs_0$ is the optimal cost of \eqref{eq:lqr1}, realized by the state feedback law $\us = K_\infty \xs$; and \label{eq:nominal:first:sub2}
    \item $P_\infty$ is the unique solution to \eqref{eq:nominal_riccati}. \label{eq:nominal:first:sub3}
  \end{inlinelist}

  Since there exists a stabilizing controller $\us = K \xs$, the optimal cost of \eqref{eq:lqr1} is bounded above. Indeed, we can write the closed-loop stage cost at time $k$ equivalently as
  \begin{equation}\label{eq:cl-stage-cost}
    \Es{\muls}{\probmul}{\infty} \left[ \trans{\xs_k} (Q + \trans{K} R K) \xs_k \right] = \tr \left( (Q + \trans{K} R K) \Es{\muls}{\probmul}{\infty} \left[ \xs_k \trans{\xs_k}  \right] \right)
  \end{equation}
  and by definition of \emss{}, $\Es{\muls}{\probmul}{k} \left[ \xs_k \trans{\xs_k} \right] \leq c\gamma^k \nrm{\xs_0}$, $\forall k \in \N$, with $c >0$ and $\gamma \in (0, 1)$, which implies that \eqref{eq:cl-stage-cost} is summable. Furthermore, positive definiteness of $Q$ and $R$ guarantees that $\{P_k\}_{k\in \posN}$ is a monotone increasing sequence (\ie{} $P_{k+1} \sgeq P_{k}$, $\forall k \in \posN$) of positive definite matrices if $P_0 = 0$. Therefore the sequence $\{P_k\}_{k \in \N}$ converges to some $P_{\infty} \sgt 0$, proving \ref{eq:nominal:first:sub1}. 

  Next we prove that $K_\infty$ is the optimal controller. To do so, we write the Riccati equation for states $x_k$ and $x_{k+1}$ at subsequent time steps as (see \eqref{eq:bellman3}):
  \begin{equation} \label{eq:telescope}
    \begin{split}
      \trans{\xs_k} P_\infty \xs_k &= \trans{\xs_k} (Q + \trans{K_\infty} R K_\infty)\xs_k + \Er{\muls_0}{\probmuls} \left[\trans{\xs_{k+1}}P_\infty \xs_{k+1} \mid \xs_k \right] \\
      \trans{\xs_{k+1}} P_\infty \xs_{k+1} &= \trans{\xs_{k+1}} (Q + \trans{K_\infty} R K_\infty)\xs_{k+1} + \Er{\muls_1}{\probmuls} \left[\trans{\xs_{k+2}}P_\infty \xs_{k+2} \mid \xs_{k+1} \right],
    \end{split}
  \end{equation}
  where $\xs_{k+1} = A(\muls_k)\xs_k + B(\muls_k)K_\infty \xs_k$. Taking the expected value of both sides of the equalities in \eqref{eq:telescope}, noting that $\Es{\muls}{\probmuls}{N} \left[\Er{\muls_0}{\probmuls} \left[\trans{x_{k+1}}P_\infty x_{k+1} \mid x_k \right]\right] = \Es{\muls}{\probmuls}{N} \left[ \trans{x_{k+1}} P_\infty \xs_{k+1}\right]$ and summing the equalities for $k=0, \dots, \partsim-1$ allows us to write:
  \begin{equation} \label{eq:bndsvaluefct}
    \begin{split}
      \trans{\xs_0}P_{\partsim{}}\xs_0 &\leq \Es{\muls}{\probmul}{\partsim} \ssum_{k=0}^{\partsim{}-1}  \left[ \trans{\xs_k} (Q + \trans{K_{\infty}} R K_{\infty}) \xs_k \right] = \trans{\xs_0} P_{\infty} \xs_0 - \Es{\muls}{\probmul}{\infty} \left[ \trans{x_{\partsim{}}} P_{\infty} \xs_{\partsim{}} \right]\\
      &\leq \trans{\xs_0} P_{\infty} \xs_0, \quad \forall \partsim{} \in \N, \forall \xs_0 \in \Re^\xd,
    \end{split}
  \end{equation}
  with $x_{\partsim{}}$ produced by running \eqref{eq:sys1} for $\us_k = K_{\infty}\xs_k$ for $\partsim{}$ time steps, starting from some $\xs_0$. The first inequality in \eqref{eq:bndsvaluefct} follows from dynamic programming, since we know that for the finite horizon case $P_{\partsim{}}$ describes the optimal cost and the final inequality follows from $P_{\infty} \sgt 0$. Taking the limit of $\partsim{} \rightarrow \infty$ we have that $\trans{\xs_0} P_{\infty} \xs_0$ is the closed-loop cost for $ \us_k = K_{\infty}\xs_k$ and that $P_{\infty}$ and $K_{\infty}$ describe the optimal cost and the optimal controller respectively, proving \ref{eq:nominal:first:sub2}. 

  Notice that \eqref{eq:bellman3} implies the Lyapunov condition \eqref{eq:lyapunov} since $R \sgt 0$ and $Q \sgt 0$. This holds for any solution of the Riccati equation \eqref{eq:nominal_riccati}, proving \ref{the:nominal:third}. 

  Note that by \eqref{eq:bellman4}, $P_{\infty}$ satisfies the Riccati equation \eqref{eq:nominal_riccati}. Since \eqref{eq:bndsvaluefct} holds for any $\tilde{P}_\infty$ that satisfies \eqref{eq:nominal_riccati} we can see that $P_\infty \sleq \tilde{P}_\infty$ when we take $\partsim{} \rightarrow \infty$, since $\Es{\muls}{\probmul}{\partsim}[ \tilde{\xs}_\partsim \tilde{P}_\infty \xs_\partsim ] \rightarrow 0$ by definition of \mss{}. We use \eqref{eq:bellman3} to write:
  \begin{equation} \label{eq:telescope2}
    \begin{split}
      \trans{\xs_k} \tilde{P}_\infty \xs_k &\leq \trans{\xs_k} (Q + \trans{K_k} R K_k)\xs_k + \Er{\muls_0}{\probmuls} \left[\trans{\xs_{k+1}}\tilde{P}_\infty \xs_{k+1} \mid \xs_k \right],
    \end{split}
  \end{equation}
  where $\xs_{k+1} = A(\muls_k)\xs_k + B(\muls_k)K_k \xs_k$. The inequalities follow from the fact that we did not use the optimal controller $\tilde{K}_\infty = -(R + G(\tilde{P}_\infty))^{-1} H(\tilde{P}_\infty)$, instead we use $K_k$, which denote the finite horizon optimal controllers. Then summing \eqref{eq:telescope2} similarly to what we did for \eqref{eq:telescope} results in:
  \begin{equation} \label{eq:bndsvaluefct2}
    \trans{\xs_0}\tilde{P}_\infty \xs_0 \leq \trans{\xs_0} P_\partsim \xs_0 + \Es{\muls}{\probmul}{\infty} \left[ \tilde{\xs}_\partsim \tilde{P}_\infty \xs_\partsim \right]
  \end{equation}
  where $\tilde{\xs}_{k+1} = A(\muls_k) + B(\muls_k)\tilde{K}_\infty$. Then taking $\partsim \rightarrow \infty$ in \eqref{eq:bndsvaluefct2} implies $\Es{\muls}{\probmul}{\partsim}[ \tilde{\xs}_\partsim \tilde{P}_\infty \xs_\partsim ] \rightarrow 0$ as before. So $\tilde{P}_\infty \sleq {P}_\infty$ and we know $P_\infty \sleq \tilde{P}_\infty$ from earlier. Therefore ${P}_\infty = \tilde{P}_\infty$, proving \ref{eq:nominal:first:sub3}.

  We can use the arguments by \cite{Balakrishnan2003} to show that the solution \eqref{eq:nominal_riccati} is obtained by solving \eqref{eq:sdp}, proving \ref{the:nominal:second}. This follows from the complementary slack condition.

  \hfill \QED


\section{Proofs for distributionally robust case}
\subsection{Proof of Theorem~\ref{the:drlyapunov}} \label{app:drlyapunov}
  The Lyapunov inequality \eqref{eq:drlyapunov} directly implies
  \begin{equation}
    P - \Ere{\mulse}{\probmule} \left[ \trans{\hat{A}(\mulse)}P\hat{A}(\mulse) \right]\sgt 0, \quad \forall \probmule \in \amb.
  \end{equation}
  Since $\probmuls ( \probmul \in \amb ) \geq 1-\beta$ we have that \eqref{eq:lyapunov} holds with probability at least $1-\beta$, proving the required result.  \hfill \QED 

\subsection{Proof of Proposition~\ref{the:ukvar}} \label{app:ukvar}
  The Bellman operator associated with \eqref{eq:lqr2} is given by:
    \begin{equation} \label{eq:bellman2}
      (T\Val_{k+1})(\xs) = \min_{\us} \left\{ \trans{\xs}Q\xs + \trans{\us}R\us + \max_{\probmule \in \amb_\covar} \Ere{\mulse}{\probmule} \left[ \Val_k(A(\mulse)x + B(\mulse)u) \right]\right\}.
    \end{equation}
    We can evaluate the expectation as in the proof of Theorem~\ref{the:lyapunov}, resulting in a similar statement to \eqref{eq:bellman3}. After extracting the terms that are independent of $\covar$ from the maximum and grouping the remaining ones together, only the following needs to be evaluated:
    \begin{equation*}
      \max_{\covar \sleq \rcovar\hat{\covar}} \trans{(\vec{A}x + \vec{B}u)} (\covar \otimes P) {(\vec{A}x + \vec{B}u)}.
    \end{equation*}
    This is equivalent to:
    \begin{equation*}
      \max_{\covar \sleq \rcovar\hat{\covar}} \tr \left( 
        \smallmat{
          \trans{z}_{1} P z_{1} & \ldots & \trans{z}_{1} P z_{\muld} \\
          \vdots & \ddots & \vdots \\
          \trans{z}_{\muld} P z_{1} & \ldots & \trans{z}_{\muld} P z_{\muld} 
        }\smallmat{
          \covar_{11} & \ldots & \covar_{1\muld} \\
          \vdots & \ddots & \vdots \\
          \covar_{\muld1} & \ldots & \covar_{\muld\muld}
        }
      \right) = \max_{\covar \sleq \rcovar\hat{\covar}} \tr \left(X \covar\right),
    \end{equation*}
    where $z_i = A_ix + B_iu$. Since $X = \trans{Z} P Z \sgt 0$, with $Z = \begin{bmatrix} z_1 & \ldots & z_\muld \end{bmatrix}$. The maximum corresponds to a support function for which one can easily verify using the optimality conditions that:
    \begin{equation*}
      \max_{\covar \sleq \rcovar\hat{\covar}} \tr \left(X \covar\right) = \tr \left(\rcovar X\hat{\covar}\right) = \rcovar \trans{(\vec{A}x + \vec{B}u)} (\hat{\covar} \otimes P) \trans{(\vec{A}x + \vec{B}u)}.
    \end{equation*}
    As such the Bellman operator is still of a similar form to \eqref{eq:bellman3}. Therefore the remaining arguments from the proof of Proposition~\ref{the:nominal} are all applicable using $\rcovar \hat{\covar}$ instead of $\covar$. More specifically, we know that the cost is bounded above since there exists a stabilizing $K$ for $\rcovar \hat{\Sigma}$ by assumption. This means that value iteration will converge to some $P_{\infty}$. We can once again telescope the Riccati equation and prove that this $P_\infty$ and $K_\infty$ are optimal and the unique solution to the Riccati equation. The arguments of \citep{Balakrishnan2003} are still applicable as well. 

    We still need to prove that the resulting controller stabilizes \eqref{eq:sys1} for the true $\covar$ with probability at least $1-\beta$. For this consider the Riccati equation \eqref{eq:nominal_riccati} which is equivalent to:
    \begin{equation*}
      \begin{split}
        P &- \max_{\covar \sleq \rcovar\hat{\covar}} \trans{(\vec{A}_0 + \vec{B}_0\dd{K}_\infty)}(\covare \otimes P)(\vec{A}_0 + \vec{B}_0\dd{K}_\infty)
        = Q + \trans{K_\infty}RK_\infty \sgt 0,
      \end{split}
    \end{equation*}
    Due to Theorem~\ref{the:drlyapunov}, $\dd{K}_\infty$ then stabilizes \eqref{eq:sys1} in the mean square sense with probability at least $1-\beta$. 
    
    \hfill \QED

\subsection{Proof of Theorem~\ref{the:ukmean}} \label{app:ukmean}
We need to prove the following statements:
\begin{inlinelist}
  \item \label{item:constraint} constraints \eqref{seq:mlmi2} and \eqref{seq:mlmi3} are equivalent to the constraint in \eqref{eq:quadraticbound},
  \item \label{item:upperboundcost} the solution of \eqref{eq:quadraticbound} upper bounds the true optimal cost of \eqref{eq:lqr2}
  \item \label{item:quadraticbound} the cost of \eqref{eq:lmiukm2} is equivalent to that of \eqref{eq:quadraticbound}
  \item \label{item:mss} the resulting controller is mean square stabilizing for \eqref{eq:sys1}.  
\end{inlinelist}

Using the quadratic parametrization of $V(\xs)$ and applying the same tricks as in the proof of Proposition~\ref{the:ukvar} we can rewrite the constraint of \eqref{eq:quadraticbound} as follows:
\begin{equation*}
  \begin{split}
    P - Q - \trans{K}RK &- \trans{(\vec{A} + \vec{B}K)} (\covardr \otimes P) (\vec{A} + \vec{B}K) \\
    &- \trans{(A(\mean) + B(\mean)K)} P (A(\mean + B(\mean)K) \sgeq 0, \quad \forall \mean \in \mathcal{D},
  \end{split}
\end{equation*}
where we used $\us = K \xs$ and define $\mathcal{D} \dfn \left\{ \trans{(\mean - \hat{\mean})} \hat{\covar}^{-1} (\mean - \hat{\mean}) \leq \rmean^2 \right\}$. Pre- and post-multiplying by $W = P^{-1}$ and setting $V = KW$ results in:
\begin{equation*}
  \begin{split}
    W - \trans{W} Q W - \trans{V} R V &- \trans{(\vec{A}W + \vec{B}V)}(\covardr \otimes W)(\vec{A}W + \vec{B}V) \\
    &- \trans{(A(\mean)W + B(\mean)V)}W^{-1}(A(\mean)W + B(\mean)V), \quad \forall \mean \in \mathcal{D}.
  \end{split}
\end{equation*} 
Applying Schur's complement lemma \citep[Sec. 2.1]{Boyd1994} then gives the following:
\begin{equation} \label{eq:lmiukm1}
  \smallmat{
    W & \star & \star & \star & \star \\
    \vec{A}W + \vec{B}V & \covardr^{-1} \otimes W & & & \\
    (\hat{A}W + \hat{B}V) + \ssum_{i=1}^\muld \mean_i F_i & & W & & \\
    Q^{\frac{1}{2}}W & & & I_\xd & \\
    R^{\frac{1}{2}}V & & & & I_\ud 
  } \sgeq 0, \quad \forall \tilde{\mean} \in \tilde{\mathcal{D}},
\end{equation}
where $\tilde{\mathcal{D}} \dfn \left\{ \trans{\tilde{\mean}} \hat{\covar}^{-1} \tilde{\mean} \leq \rmean^2 \right\}$ and $F_i = A_iW + B_iV$. Defining $\zeta = \hat{\covar}^{\nicefrac{-1}{2}} \tilde{\mean}$ allows us to write $\ssum_{i=1}^\muld \mean_i F_i = \ssum_{i=1}^\muld \zeta_i H_i$. Note that, after using this equality, \eqref{eq:lmiukm1} is of the form 
\begin{equation*}
  U_0 + \ssum_{i=1}^\muld \delta_i U_i \sgeq 0, \quad  \forall \delta \in \left\{ \delta \mid \nrm{\delta}_2 \leq \rho \right\}.
\end{equation*}
Hence we can apply a slightly modified version of \cite[Theorem 6.2.1]{Ben-Tal2000} where we exploit the sparsity of $U_i$, which results in conditions \eqref{seq:mlmi2} and \eqref{seq:mlmi3}. More specifically we have:
\begingroup
\allowdisplaybreaks
\begin{subequations}
  \begin{alignat}{10}
    \trans{\xi}(U_0 + \ssum_{i=1}^\muld \delta_i U_i) \xi &= \trans{\xi} U_0 \xi + 2 \ssum_{i=1}^\muld \zeta_i \trans{\xi}_3  H_i \xi_1 \nonumber \\ 
    & \qquad {\scriptstyle \left( \xi =\trans{\left[\begin{smallmatrix} \trans{\xi}_1 & \trans{\xi}_2 &\trans{\xi}_3 &\trans{\xi}_4 &\trans{\xi}_5 \end{smallmatrix}\right]} \right)} \nonumber \\
    & = \trans{\xi} U_0 \xi + 2 \trans{\xi_L} \left[ \ssum_{i=1}^\muld \zeta_i L^{\nicefrac{-1}{2}} H_i S^{\nicefrac{-1}{2}} \xi_S \right] \nonumber \\
    & \qquad {\scriptstyle \left( \xi_S = S^{\nicefrac{1}{2}} \xi_1, \, \xi_L = L^{\nicefrac{1}{2}} \xi_3 \right)} \nonumber \\
    & \geq \trans{\xi} U_0 \xi - 2 \nrm{\xi_L}_2 \left[ \ssum_{i=1}^\muld \abs{\zeta_i} \nrm{ L^{\nicefrac{-1}{2}} H_i S^{\nicefrac{-1}{2}} \xi_S}_2 \right]  \nonumber \\
    & \geq \trans{\xi} U_0 \xi - 2 \nrm{\xi_L}_2 \sqrt{ \ssum_{i=1}^\muld \rmean^2 \nrm{ L^{\nicefrac{-1}{2}} H_i S^{\nicefrac{-1}{2}} \xi_S}_2^2 } \label{eq:btderiv1} \\
    & = \trans{\xi} U_0 \xi - 2 \nrm{\xi_L}_2 \sqrt{ \rmean^2 \trans{\xi}_S \left[ \ssum_{i=1}^\muld S^{\nicefrac{-1}{2}} \trans{H}_i L^{-1} H_i S^{\nicefrac{-1}{2} }\right] \xi_S }  \nonumber \\
    & \geq \trans{\xi} U_0 \xi - 2 \nrm{\xi_L}_2 \nrm{\xi_S}_2 \label{eq:btderiv2},
  \end{alignat}
\end{subequations}
\endgroup
where we used $\nrm{\zeta} \leq \rmean$ for \eqref{eq:btderiv1} and \eqref{eq:btderiv2} holds when:
\begin{equation*}
  \ssum_{i=1}^\muld S^{\nicefrac{-1}{2}} (\rmean \trans{H}_i) L^{-1} (H_i \rmean) S^{\nicefrac{-1}{2} } \leq I_{n_w},
\end{equation*}
which after pre- and post-multiplying by $S^{\nicefrac{1}{2}}$ and applying Schur's complement lemma \citep[Sec. 2.1]{Boyd1994} is shown to be equivalent to \eqref{seq:mlmi2}. The result then holds when:
\begin{equation} \label{eq:fcond}
  \begin{split}
    \trans{\xi}U_0 \xi &\geq 2 \sqrt{(\trans{\xi}_L \xi_L) (\trans{\xi}_S \xi_S)},
  \end{split}
\end{equation}
for which \eqref{seq:mlmi3} is a sufficient condition since:
\begin{equation*}
  \begin{split}
    \sqrt{2} (\trans{\xi}_L \xi_L +  \trans{\xi}_S \xi_S) &\geq 2 \sqrt{(\trans{\xi}_L \xi_L) (\trans{\xi}_S \xi_S)} \\
    2 \trans{\xi}_L \xi_L + 4(\trans{\xi}_L \xi_L) (\trans{\xi}_S \xi_S) + 2 \trans{\xi}_S \xi_S &\geq 4(\trans{\xi}_L \xi_L) (\trans{\xi}_S \xi_S).
  \end{split}
\end{equation*}
So therefore \eqref{seq:mlmi3}, which can be written as $\trans{\xi}U_0 \xi \geq \sqrt{2} (\trans{\xi}_L \xi_L +  \trans{\xi}_S \xi_S)$, implies \eqref{eq:fcond}, proving \ref{item:constraint}.

Next we will prove that the result of \eqref{eq:quadraticbound} upper bounds the true cost of \eqref{eq:lqr2}. We introduce $\ell(x) = \trans{x}(Q + \trans{K}RK)x$ and $A_K(\muls) = A(\muls) + B(\muls)K$. Then we can use the constraint in \eqref{eq:quadraticbound} to write
\begin{equation*}
  \begin{split}
    \trans{\xinitsr} P \xinits &\geq \ell(\xinitsr) + \max_{\probmule \in \amb} \Ere{\mulse}{\probmule} \left[ \trans{\xinitsr} \trans{A_K(\mulse_0)} P A_K(\mulse_0) \xinitsr \right], \\
    & \geq \max_{\probmule \in \amb} \Ese{\mulse}{\probmule}{2} \left[ \ell(\xinitsr) + \ell( A_K(\mulse_0) \xinitsr) + \trans{\xinitsr} \trans{A_K(\mulse_0)} \trans{A_K(\mulse_1)} P A_K(\mulse_1) A_K(\mulse_0) \xinitsr\right]
  \end{split}
\end{equation*}
which can be recursively applied to show that 
\begin{equation} \label{eq:upperboundcost}
  \trans{\xinitsr} P \xinitsr \geq \max_{\probmule \in \amb} \Ese{\mulse}{\probmule}{\infty} \left[\ssum_{k=0}^{\infty} \trans{\xs_k}(Q + \trans{K}RK)\xs_k\right].
\end{equation}
\sloppy We can then take the expectation with respect to $\xinitsr$ of both sides. Noting that $\Er{\xinitsr}{\probxinit} \left[ \trans{\xinitsr} P \xinitsr \right] = \tr(P \Er{\xinitsr}{\probxinit} \left[ \xinitsr \trans{\xinitsr} \right]) = \tr(PI_{n_x})$ explains why the trace of $W^{-1}$ is maximized in \eqref{eq:lmiukm2}, proving \ref{item:upperboundcost}. Finally consider the constraint of \eqref{eq:quadraticbound}, which is a DR Lyapunov decrease condition. As such we can use Theorem~\ref{the:drlyapunov}, proving \ref{item:quadraticbound}. \hfill \QED
}{}
\end{document}